\documentclass[%
 aps,
 amsmath,amssymb,
 aps,
]{revtex4}
\usepackage{amssymb}
\usepackage{amsfonts}
\usepackage{amsmath}
\usepackage{epsfig, color}

\usepackage{graphicx}
\graphicspath{%
    {converted_graphics/}
    {/}
}
\usepackage{comment}
\usepackage{amsfonts,amssymb,stmaryrd,amsthm}
\usepackage{amsmath}
\usepackage{graphicx}
\usepackage{enumerate}
\usepackage{subfigure}
\usepackage{amscd}
\usepackage{pb-diagram}
\usepackage[toc]{appendix}
\usepackage{afterpage}
\usepackage{bm}

\newcommand{\eproof}{\hfill{\vrule height5pt width5pt depth0pt}\medskip}

\newcommand{\R}{{\mathbb R}}

\newtheorem{theorem}{Theorem}[section]
\newtheorem{lemma}[theorem]{Lemma}

\begin{document}

\title{Traveling wave solutions to Kawahara and related equations}


\author{Stefan C. Mancas}
\email[Electronic address for correspondence: ]{mancass@erau.edu}

\affiliation{Department of Mathematics, Embry-Riddle Aeronautical University,\\ Daytona Beach, FL. 32114-3900, U.S.A.}


\begin{abstract}
Traveling wave solutions to  Kawahara equation (KE), transmission line (TL), and Korteweg-de Vries (KdV) equation are found  by using an elliptic function method  which is more general than the  $\mathrm{tanh}$-method. The method works by assuming  that a polynomial ansatz  satisfies a Weierstrass equation, and  has  two advantages: first, it reduces the number of terms in the  ansatz by an order of two, and second, it uses Weierstrass  functions  which satisfy an  elliptic equation for the  dependent variable instead of the  hyperbolic tangent functions which only  satisfy  the  Riccati equation with constant coefficients.  

When the polynomial ansatz in the traveling wave variable  is of first order, the equation reduces to the  KdV equation with only a cubic dispersion term, while for the KE which  includes a  fifth order dispersion term the polynomial  ansatz must necessary be of quadratic type. 

By solving the  elliptic equation with coefficients that depend on the boundary conditions, velocity of the traveling waves,  nonlinear  strength, and dispersion coefficients, in the case of  KdV equation we find the well-known solitary waves (solitons) for zero boundary conditions, as well as  wave-trains of  cnoidal waves for nonzero boundary conditions. Both solutions are either compressive  (bright)  or rarefactive  (dark), and  either propagate to the left or right with arbitrary  velocity.

In the case of KE  with nonzero boundary conditions and zero cubic dispersion, we  obtain   cnoidal  wave-trains  which  represent   solutions to the TL equation.  For KE with  zero boundary conditions and all the dispersion terms present, we  obtain again solitary waves, while for  KE with all coefficients present and  nonzero boundary condition, the solutions  are written in terms of Weierstrass elliptic functions. For all cases of the KE  we only find  bright waves that  are propagating to the  right  with velocity that is a function of both dispersion coefficients.

\end{abstract}
\maketitle

\section{Introduction}
In recent years many  methods  have been used  to find analytic solutions to nonlinear partial differential equations (PDEs). 
Among the multitude of  papers, we shall only refer to two sets of studies which will pertain to this work. The first class of papers  are:  truncation procedure in the Painlev\'{e} analysis \cite{Weiss} in which  authors  define the Painlev\'{e} property that determines the Lax pairs of the Burgers, KdV, and the modified KdV equations;  Hirota bilinear method \cite{Hirota} where   multiple collisions of $N$  solitons with varying amplitudes  have been obtained for the KdV equation; the Prelle-Singer method \cite{Prelle} where  a system of differential equations has been shown to have an elementary integral expressible in terms of exponentials, logarithms and algebraic functions; the factorization method \cite{Rosu} where traveling wave solutions of the standard and compound KdV-Burgers equations are found using factorizations; or the  homogeneous balance method \cite{Wang} where  solitary wave solutions of two types of variant Boussinesq equations are obtained. Then, we can also enumerate  the trial function method \cite{Kud6} where transformations of solutions obtained by the Weiss-Tabor-Carnevale method are used for investigation of  Kuramoto-Sivashinsky equation;  the nonlinear transformation method \cite{Otw} where the authors  constructed traveling wave solutions for nonlinear diffusion equations with polynomial nonlinearities;  the well-known inverse scattering transform \cite{Newell,Ablo} and  B\"acklund transformation \cite{Miura}; the first integral method  \cite{Nizo} where Nizovtseva uses a  first integral method which gives  singular and kink profiles for the Allen-Cahn hyperbolic equation. 

The second class of papers  mentioned here are: the simplest equation method \cite{Kud5} where Kudrsyashov  uses the general solutions of   simplest nonlinear differential equations and takes  into consideration all possible singularities of Kuramoto- Sivashinsky equation, as well as  the equation for description of nonlinear waves in a convective fluid; the $G'/G$ expansion method \cite{Kud7} where it is shown to be equivalent to the $\tanh$-method first developed by Malfliet and  Hereman \cite{Malf1,Malf2,Malf3}; the automatic method of Parkes \cite{Parkes}; the method of $Q$ functions  \cite{Kud1,Kud8} used on Fisher equation and on a seventh order ODE;  the generalized Riccati equation method \cite{Yan} where a new generalized transformation is  applied to  Whitham-Broer-Kaup (WBK) equation; the  $\sinh$-$\cosh$-method \cite{Was} where the author finds solitons, kinks, and periodic solutions of Benjamin-Bona-Mahony (BBM) equation;  the modified $\tanh$-method \cite{Fan1,Fan2} where the author uses a modified by a parameter Riccati equation; the algebraic method where algorithms using sophisticated  Mathematica programs are used to find closed-form solutions in terms of Jacobi elliptic functions \cite{Bald,Her}; the $\mathrm{sech}$-method \cite{Ma} to find solitons to a seventh order KdV equation. Then we also include the Jacobi elliptic function method \cite{Fan3} used on a double sine-Gordon, Hirota equation, and the coupled Schr\"{o}dinger-KdV system; the work of  Fu and Liu \cite{Fu,Liu} on Jacobi elliptic function expansion method; and  Porubov \cite{Por1,Por2,Por3} on  traveling  periodic solutions of a pair of coupled nonlinear Schr\"{o}dinger 
equations  obtained in terms of  Weierstrass elliptic $\wp$ functions.

More importantly, if for the former class of papers the methods yield   restrictive solutions involving elementary functions which generate solitary waves, singular solutions as rational functions, periodic trigonometric solutions, kinks and fronts, the latter studies involve finding analytical solutions of evolution equations in terms of Jacobi, Weierstrass or  elliptic theta functions. 

Motivated by the work of the authors of the second class of papers, an elliptic function method, which is easier to implement and more
general than the hyperbolic tangent method  is applied to a  nonlinear  dispersive PDE known as Kawahara equation (KE) to find periodic solutions in terms of Weierstrass  $\wp$ elliptic functions, Jacobi elliptic or hyperbolic functions. This equation  takes the form 
\begin{equation}\label{eq:rosenau}
u_{t}+\kappa uu_{x}+\alpha u_{xxx}-\beta u_{xxxxx}=0,
\end{equation}
and was investigated numerically in a study of  magneto-acoustic  waves in a cold collision-free plasma \cite{Kawa5}. The coefficients of this equations are: 
 $\alpha, \beta$ - third and fifth order dispersive terms,  $\kappa$ - the strength of the nonlinearity (wave steepening) and are real constants. We may assume that $\beta>0$,  because by the transformations $u\rightarrow -u,$ $x\rightarrow -x,$ and $t\rightarrow -t,$ we obtain the same equation as (\ref{eq:rosenau}) with the dispersive terms reversed in sign \cite{Kawa5}. Under certain circumstances, the third order dispersion coefficient $\alpha$ becomes very small, or even zero,  so one should include the higher order dispersion $\beta$  which will balance the nonlinear effect $\kappa$ \cite{Kuka,Has}.

We can write Eq. (\ref{eq:rosenau}) in Hamiltonian form \cite{Drazin}
\begin{equation}\label{ea2}
u_t=\frac{\partial}{\partial x}\left(\frac{\delta \mathcal H}{\delta u}\right),
\end{equation}
where the Hamiltonian is
\begin{equation}
\mathcal H=-\frac{1}{6}\kappa u^3+\frac 1 2 \alpha {u_x}^2+\frac 12 \beta {u_{xx}}^2.
\end{equation}
Using this Hamiltonian, KE  (\ref{eq:rosenau}) has the conserved energy density
\begin{equation}
\mathcal I=\int_{-\infty}^{\infty}\mathcal H dx.
\end{equation}
By using the Fr\'{e}chet derivative which corresponds to the Euler-Lagrange operator 
\begin{equation}
\frac{\delta }{\delta u}=\frac{\partial  }{\partial u}-\frac{d}{dx}\frac{\partial }{\partial u_x}+\frac{d^2}{dx^2}\frac{\partial }{\partial u_{xx}},
\end{equation}
then Eq. (\ref{ea2}) becomes  KE (\ref{eq:rosenau}).

By applying a  traveling wave ansatz  $u(\xi)=u(x-ct)$,  with $c$ being the velocity of the unidirectional traveling wave in the $x$ direction at time $t$,  yields a fifth order ordinary differential equation (ODE) in the traveling wave variable $\xi$
\begin{equation}\label{eq0}
-cu_{\xi} +\kappa u u_\xi+\alpha u_{\xi \xi \xi} -\beta u_{\xi \xi \xi \xi \xi}  =0.
\end{equation}
By one integration this reduces  to the fourth order ODE
\begin{equation}\label{eq1aa}
-cu+\frac \kappa 2 u^2 +\alpha u_{\xi \xi} -\beta u_{\xi \xi \xi \xi}  =\mathcal{A}
\end{equation}
with $\mathcal{A}$ an arbitrary  integration constant which can be zero or not depending on the types of boundary conditions chosen.
By multiplying by $u_\xi$ and integrating once  we obtain a conserved quantity  for Eq. (\ref{eq:rosenau}) in the traveling wave variable $\xi$
\begin{equation}\label{lab}
\mathcal C=-2\mathcal A u-cu^2+\frac{\kappa}{3}u^3+\alpha {u_\xi}^2 -\beta \left[2 u_\xi u_{\xi\xi\xi}-(u_{\xi\xi})^2\right]\equiv const.
\end{equation}

In his comment to Assas\rq{} paper \cite{Assa10},  Kudryashov developed the solutions of KE using the tanh-method \cite{Kud3,Kud2}. This method was originally used by  Malfliet and Hereman \cite{Malf1,Malf2,Malf3} and has the advantage  of reducing nonlinear ODEs  into systems of algebraic equations, that might be easier to solve. Kudryashov explained that Eq. (\ref{eq1aa}) does not pass the Painlev\'{e} test, but nevertheless one can find solitary waves of higher order by writing the Laurent series expansion for a function $Y(\xi) $ which must include a pole of order four  \cite{Kud1,Kud4}, where the function in the expansion solves the Riccati equation with constant coefficients 
\begin{equation}\label{ric}
Y_{\xi}=\eta (1- Y^2)
\end{equation}
with solution $Y(\xi) =\mathrm{tanh} (\eta \xi)$. 
Therefore, if we assume solutions of the form 
\begin{equation}\label{tan}
u(\xi)=\sum_{i=0}^n C_i Y^i,
\end{equation}
once the numbers of terms $n$ is determined using the balancing principle \cite{Nick}, we can  write the  solutions of Eq. (\ref{eq1aa}) in terms of the solutions of the Riccati equation (\ref{ric}). Since the hyperbolic tangent solution is a particular solution of the Riccati equation, and any other solution can be found using the transformation $Y=\mathrm{tanh} (\eta \xi)+\frac{1}{W}$, where $W$ satisfies a first order linear equation,  then all the other solutions which are meromorphic to $u$ can be written  using a new expansion  in $W$ with the same number of terms. More than that, all forms of the general solution of the Riccati equation have the same Laurent series  and they differ only  by arbitrary constants \cite{Ere}.  Therefore, using different ODEs  as  generators of particular solutions, one can find a rich class of meromorphic solutions to evolution equations, which  unite  many approaches  involving  elementary functions. These methods are not restrictive to only parabolic equations ($d/dt \rightarrow -c~ d/d \xi$), as they were also successfully implemented to find solutions to hyperbolic PDEs  ($d^2/dt^2 \rightarrow c^2~ d^2/d {\xi}^2$) such as  Boussinesq  \cite{Wang} and improved Boussinesq equation \cite{Kud2,Ablo}, Klein-Gordon \cite{Zhang}, and  Allen-Cahn equation via the first integral method  \cite{Nizo}. Note that the elliptic function method will not work if the ODE contains both even and odd derivative terms, see Lemma V.1.  in the Appendix, for the explanation.

In order to determine the number of terms in the expansion of the ansatz, we compute the second order derivative $d^2/d \xi^2$ for which  the leading term is $\eta^2 (1-Y^2)^2 ~d^2/d Y^2$, while for the fourth order derivative $d^4/d \xi^4$ the leading term is $\eta^4 (1-Y^2)^4~d^4/d Y^4$. Thus, when we balance the nonlinear term with the higher order derivatives, we must distinguish between two different cases. When $\beta =0$ (KdV equation), we only need to  balance  $u_{\xi\xi}$ with $u^2$ which  leads to  $4+(n-2)=2n \Rightarrow n=2$. For the second case (KE), we balance $u_{\xi\xi\xi\xi}$ with $u^2$   which  leads to  $8+(n-4)=2n \Rightarrow n=4$.
Therefore, the solutions for the KdV or KE  must take the form
\begin{equation}\label{tan1}
\begin{aligned}
u(\xi)&=A_0+A_2~ \mathrm{tanh}^2 (\eta \xi), \quad \quad \beta=0\\
u(\xi)&=B_0+B_2~ \mathrm{tanh}^2 (\eta \xi)+B_4 ~\mathrm{tanh}^4 (\eta \xi), \quad \quad \beta\ne0
\end{aligned}
\end{equation}
with the first and third order coefficients identically zero \cite{Kud2}.

Since  KdV equation also possesses non elementary solutions in terms of Jacobi elliptic functions \cite{KdV} which are not solutions of the Riccati equation  (\ref{ric}), we extend the ansatz of the function $Y$ and we replace the Riccati equation  by an elliptic equation, i.e., the function $Y$ in the {\it new} ansatz  given by (\ref{tan}) satisfies
\begin{equation}\label{eq3}
{Y_\xi}^2=a_0+a_1 Y+a_2 Y^2+a_3 Y^3, \quad \quad a_3 \ne0.
\end{equation}
This new ansatz has the advantage of extending the classes of solutions to include  elliptic functions \cite{Fan3,Liu,Kud5}, and  as a special case when two of the roots of the cubic polynomial in $Y$ collide, the solitary waves can be recovered as a limit case of  cnoidal waves \cite{Man2,Man3}. 
The constants $a_i$ which depend on the system parameters $\alpha, \beta, \kappa$,  the speed $c$, and boundary conditions $\mathcal A$  respectively, can be found algebraically after the ansatz passes the balancing principle  \cite{Nick}, which will determine the number of terms $n$ in the expansion given by (\ref{tan}).

Using the new ansatz, and  by balancing, we obtain $n+1=2n\Rightarrow n=1$ when $\beta=0$ and $n+2=2n \Rightarrow n=2$ when $\beta\ne 0$. Therefore, by replacing Riccati  equation  (\ref{ric}) with the elliptic equation (\ref{eq3}), we reduce the numbers of terms in half  in the expansion of Eq. (\ref{tan}) which  is now only linear for KdV or  quadratic for KE 
\begin{equation}\label{tan2}
\begin{aligned}
u(\xi)&=A_0+A_1 Y, \quad\quad \beta=0\\
u(\xi)&=B_0+B_1 Y+B_2 Y^2, \quad\quad \beta \ne 0.
\end{aligned}
\end{equation}
%

\section{Reduction to the  KdV equation ($\beta=0$)}
When $\beta=0$, Eq. (\ref{eq:rosenau})  reduces to the  KdV equation which   describes the motion of  small amplitude and large wavelength shallow waves in dispersive systems \cite{KdV}
\begin{equation}\label{eq:KdV}
u_{t}+\kappa uu_{x}+\alpha u_{xxx}=0. 
\end{equation}  
By using $\beta =0$ in Eq. (\ref{eq1aa}), we obtain the second order ODE in $u$
\begin{equation}\label{eq1aaa}
-cu+\frac \kappa 2 u^2 +\alpha u_{\xi \xi}=\mathcal{A}.
\end{equation}
Without loss of generality we may assume that $A_0=0$, $A_1=1$ so that $u=Y$. To find the solutions of Eq. (\ref{eq1aaa}) using our procedure, we  differentiate Eq. (\ref{eq3}) to obtain  higher derivatives of $Y$
\begin{equation}\label{eqa}
\begin{aligned}
Y_{\xi\xi}&=\frac{1}{2}a_1+a_2 Y+\frac{3}{2}a_3Y^2\\
Y_{\xi\xi\xi}&=(a_2+3a_3Y)Y_{\xi}\\
Y_{\xi\xi\xi\xi}&=3 a_0a_3+\frac 1 2 a_1a_2+\left(\frac 9 2 a_1 a_3+{a_2}^2\right)Y+\frac{15}{2}a_2a_3Y^2+\frac{15}{2}{a_3}^2Y^3,
\end{aligned}
\end{equation} 
and by  substituting them into Eq. (\ref{eq:rosenau}) we obtain a  cubic polynomial in $Y$
\begin{equation}\label{eq3a}
\sum_{i=0}^3 s_i Y^i\equiv0,
\end{equation} with coefficients  $s_i$   given by the expressions
\begin{equation}\label{eq3b}
\begin{array}{l}
s_3=-\frac{15}{2}{a_3}^2 \beta\\
s_2=\frac 1 2 \left[\kappa+3 a_3(\alpha-5 a_2 \beta)\right]\\
s_1=-c+a_2\alpha-{a_2}^2\beta-\frac 92 a_1a_3 \beta\\
s_0=-\mathcal A-3 a_0a_3\beta+\frac{a_1}{2}(\alpha-a_2 \beta).
\end{array}
\end{equation}
Because all these coefficients must be zero, and since $a_3\ne 0$, from the first equation of the system ($\ref{eq3b}$), we immediately conclude that  $\beta=0$, so the reduced constants become
\begin{equation}\label{eq3bb}
\begin{array}{l}
s_2=\frac 1 2 \left(\kappa+3 a_3\alpha\right)\\
s_1=-c+a_2\alpha\\
s_0=-\mathcal  A+\frac{a_1}{2}\alpha .
\end{array}
\end{equation}
By solving  simultaneously $s_i=0$  for the coefficients $a_i$, Eq. (\ref{eq3}) becomes
\begin{equation}\label{kdv}
{Y_\xi}^2=a_0+\frac{2\mathcal A}{\alpha}Y+\frac{c}{\alpha}Y^2-\frac{\kappa}{3\alpha}Y^3 \equiv q_3(Y),
\end{equation} where $a_0$ is an arbitrary constant. According to Eq. (\ref{lab})  the conserved quantity in  $\xi$ for the KdV equation is
\begin{equation}\label{lab2}
\mathcal C=-2\mathcal A Y -cY^2+\frac{\kappa}{3}Y^3+\alpha {Y_\xi}^2  ,
\end{equation}
and by comparing Eqs. (\ref{kdv}) and (\ref{lab2}) we identify $ a_0=\frac{\mathcal C}{\alpha}$, so the arbitrary constant of the elliptic equation is proportional to the conserved quantity, and inverse proportional to the cubic dispersion.  The elliptic solutions of   Eq. (\ref{kdv}) depend on the type of roots of the cubic polynomial $q_3(Y)$, which automatically  leads to  the following sub-cases:

{\bf i)} One zero root of multiplicity two, and one simple root $Y_0=\frac{3c}{\kappa}$. This is achieved when we choose zero  boundary conditions  $\mathcal A=0$ together with $a_0=0$, i.e., the fluid is undisturbed at infinity ($Y,Y_{\xi},Y_{\xi\xi}\rightarrow 0$ as $|\xi|\rightarrow \infty$). By letting $\frac{\kappa}{3 \alpha}=\frac{1}{\sigma} $ we factor Eq. (\ref{kdv})  as
\begin{equation}\label{kdv1}
{Y_\xi}^2=\frac{1}{\sigma}Y^2(Y_0-Y),
\end{equation}
with solution the  solitary wave \cite{KdV,Rey}
\begin{equation}\label{eq00}
Y(\xi)=Y_0\mathrm{sech}^2\left[\frac 1 2 \sqrt{\frac{Y_0}{\sigma}}(\xi-\xi_0)\right]
\end{equation}
which  propagates with velocity  proportional to amplitude $Y_0$, and width inverse proportional to the square root of  $Y_0$.  

Now we must discuss the sign of $\sigma$. First possibility is  that $\sigma>0$ so  $\kappa$ and $ \alpha$ are of the same sign, and because  $Y_\xi \in \R  $ we must have  $Y_0>0$ which means that $c$ is the same sign as $\alpha$ and $ \kappa$.  These waves travel to the right ($c>0$) in a positive strength ($\kappa>0$) and  positive  dispersive medium ($\alpha>0$), and to the left  ($c<0$) in  a negative strength ($\kappa>0$) negative dispersive medium ($\alpha<0$). In both cases their amplitude is always positive and they represent the {\it positive (compressive/ bright)} solitary waves. On the other hand, if $\sigma<0$ so that $\kappa$ and $ \alpha$ are of opposite signs, then we must have  $Y_0<0$. Therefore, these are  waves that travel  to the right ($c>0$) in a negative strength ($\kappa<0$) and  positive  dispersive medium ($\alpha>0$)  medium, and to the left  ($c<0$) in  a positive strength ($\kappa>0$) and negative dispersive medium ($\alpha<0$). In both cases their amplitude is always negative  and they represent  the {\it negative (rarefactive/ dark)} solitary waves \cite{KdV,Rey}. For all other  remaining combination of coefficients, the solitons become unbounded and solutions are unphysical, when the hyperbolic secant becomes periodic with poles aligned on the real $\xi$-axis. For the solitonic case, and regardless of the signs of the coefficients, for $\xi_0=0$  and $u=Y$  the solution is
\begin{equation}\label{kdv3}
u(x,t) = \frac{3c}{\kappa}\mathrm{sech}^2\left[\frac 1 2\sqrt{\frac{c}{\alpha}}(x-ct)\right]
\end{equation}
If the velocity is fixed, the amplitude and width can be manipulated using the strength and dispersion coefficients. When $\kappa$ is big the amplitude is small, and when $\alpha$ is small the  solitons are thin, while increasing the dispersion it increases their widths when solitons spread, see Fig. \ref{fig1a} for the values of $\alpha=1$, $\kappa=1$ and  $\alpha=-5$, $\kappa=-5$ for bright solitons (top  panel), and  $\alpha=1$, $\kappa=-5$ and  $\alpha=-5$, $\kappa=1$ for dark solitons (bottom panel).

\begin{figure}[ht!]
  \begin{center}
\includegraphics[width=0.75\textwidth]{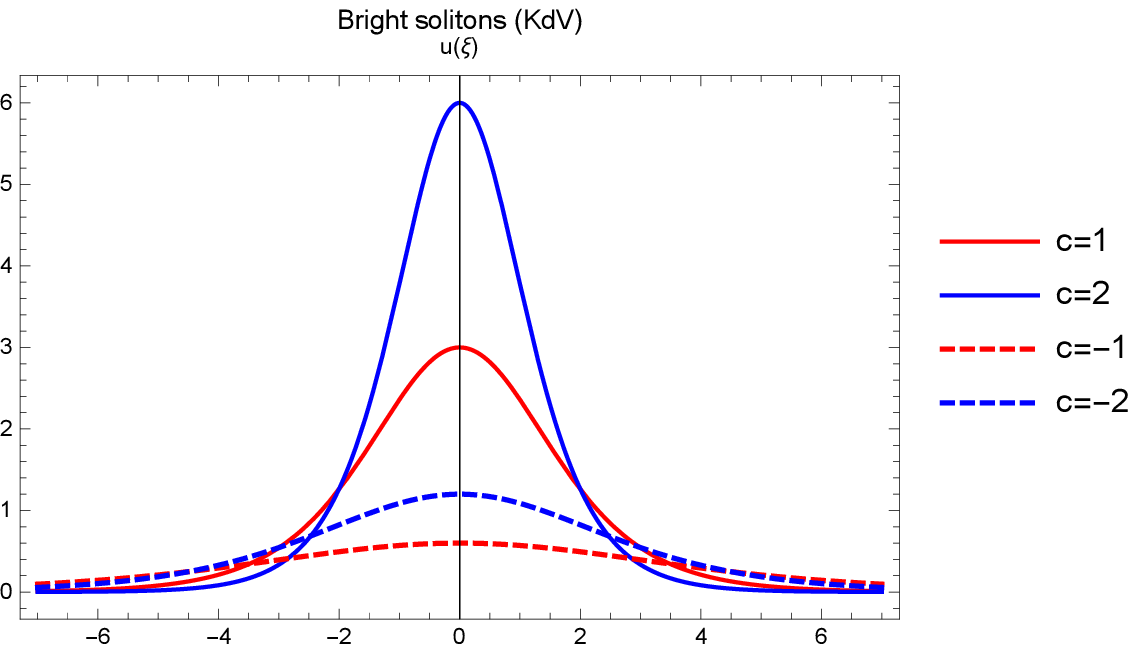}\\
\includegraphics[width=0.75\textwidth]{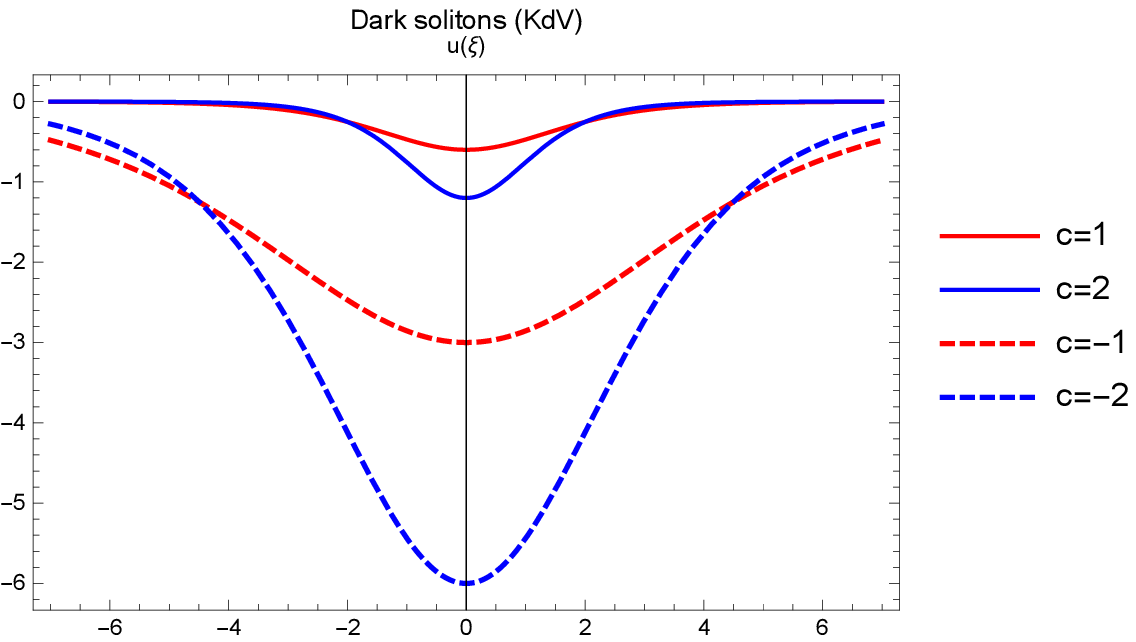}
\caption{Bright solitons for the KdV  Eq. (\ref{eq:KdV}) with zero boundary conditions $\mathcal A=0$ and $\alpha=1,~\kappa=1$ (continuous curves) when solitons propagate to the right, and $\alpha=-5,~\kappa=-5$ (dashed curves) when solitons propagate to the left (top panel). Dark solitons for the KdV  Eq. (\ref{eq:KdV}) with zero boundary conditions $\mathcal A=0$  and $\alpha=1,~\kappa=-5$ (continuous curves) when solitons propagate to the right, and $\alpha=-5,~\kappa=1$ (dashed curves) when solitons propagate to the left (bottom panel).  }
\label{fig1a}
\end{center}
\end{figure}
{\bf ii)} Now we drop the assumption that fluid must be undisturbed at infinity, so then we must have   $Y_\xi=0$ for $Y=0$ which implies that only one root is zero. This  corresponds to setting $ a_0=0$, while the other two roots are real and distinct. Under these assumptions  Eq.(\ref{kdv}) can be factored as
\begin{equation}\label{kdv04}
{Y_\xi}^2=\frac{1}{\sigma}~Y(Y_2-Y)(Y_3+Y),
\end{equation}
where $Y_{2,3}$ satisfy
\begin{equation}\label{eq4bb}
\begin{array}{l}
Y_2=\frac{Y_0\pm \sqrt{\Omega}}{2}\\
Y_3=\frac{-Y_0\pm \sqrt{\Omega}}{2}.
\end{array}
\end{equation} 
 For real roots we require  the discriminant $\Omega=\frac{9c^2+24 \mathcal A \kappa}{\kappa^2} > 0$, which restricts the values for $\mathcal A$  and $\kappa$ such that $\mathcal A \kappa>-\frac{3c^2}{8}$. The  solution of Eq. (\ref{kdv04})  is
\begin{equation}\label{can}
u(\xi)=Y_2~ \mathrm{cn}^2\left[\frac 1 2 \sqrt{\frac{1}{\sigma}(Y_2+Y_3)}(\xi-\xi_0);m\right],
\end{equation} 
which simplifies to
\begin{equation}\label{can0}
u(\xi)=\frac{Y_0\pm \sqrt{\Omega}}{2}\mathrm{cn}^2\left[\frac 1 2 \sqrt{\pm\frac{\sqrt{\Omega}}{\sigma}}(\xi-\xi_0);\sqrt{\frac 1 2 \pm \frac{Y_0}{2 \sqrt{\Omega}}}\right],
\end{equation}
where  $ \mathrm{cn}(\theta;m)$ is the Jacobian elliptic function with modulus $m=\sqrt{\frac{Y_2}{Y_2+Y_3}}$.  Using the values of the roots from system (\ref{eq4bb}), and $\xi_0=0$ the solutions of the KdV equation (\ref{eq:KdV}) with nonzero boundary conditions are 
\begin{equation}\label{Tab1}
  u(x,t) =  \begin{array}{lr}
\frac{3c \pm \sqrt{9c^2+24 \mathcal A \kappa}}{2\kappa} \mathrm{cn}^2\left[\frac 1 2 \frac{\sqrt[4]{9c^2+24 \mathcal A \kappa}}{\sqrt{\pm3\alpha}}(x-ct);\sqrt{\frac 1 2 \pm \frac{3c}{2\sqrt{9c^2+24 \mathcal A \kappa}}}\right].
 \end{array}
\end{equation}
If $\alpha, \kappa$ have same sign then $\sigma>0$, so  $Y_2+Y_3= \sqrt{\Omega}>0$, and we obtain bright cnoidal waves  which propagate to the right  ($c>0$) or to the left ($c<0$), see Fig. \ref{fig2a} (top panel)  for   right waves with $\alpha=1$, $\kappa=2$,  and left waves with $\alpha=-2$, $\kappa=-2$.   Otherwise, when $\alpha, \kappa$ have opposite  sign then $\sigma<0$, so  $Y_2+Y_3= -\sqrt{\Omega}<0$, and we obtain  dark cnoidal waves  which also  propagate to the right  ($c>0$) or to the left  ($c<0$), see Fig. \ref{fig2a} (bottom panel) for right waves with   $\alpha=-1$, $\kappa=2$, and left waves with  $\alpha=2$, $\kappa=-1$.  These solutions represent trains of periodic cnoidal waves  with shape and wavelength that depend on the amplitude of the waves. The wavelength is $\lambda=4 \sqrt{\frac{\pm \sigma}{\sqrt{\Omega}}}K(m)$, where $K(m)$ is the complete elliptic integral of first kind given by $K(m)= \int_0^{\frac{\pi}{2}}\frac{d\theta}{\sqrt{1-m^2 \sin^2\theta}}$. 
\begin{figure}[ht!]
  \begin{center}
\includegraphics[width=0.75\textwidth]{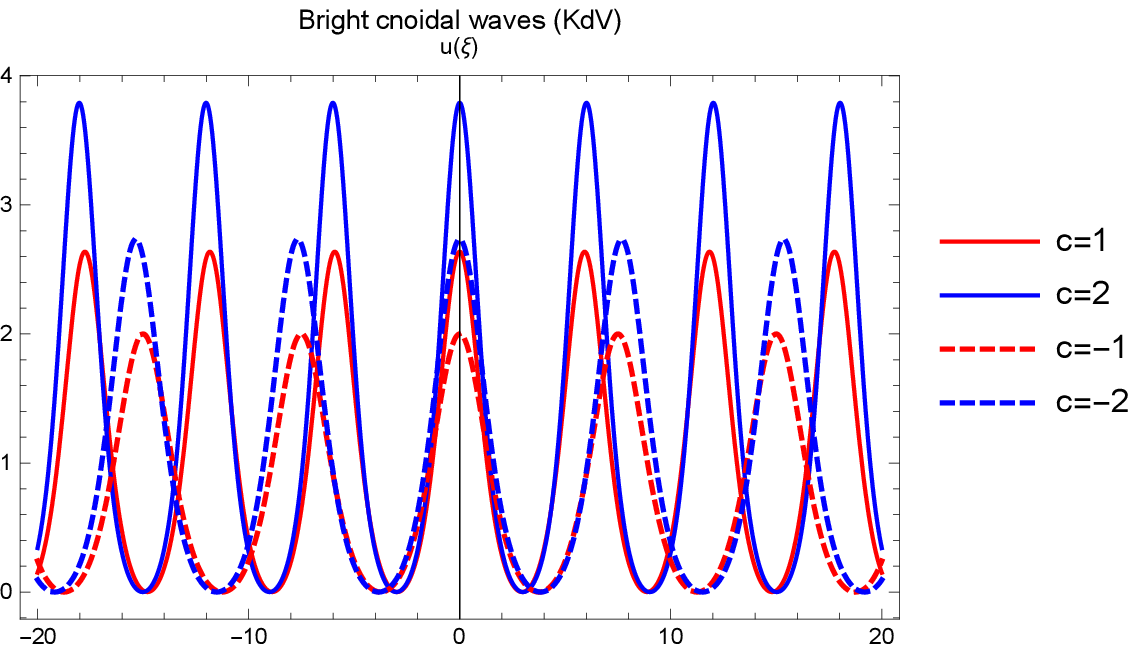}\\
\includegraphics[width=0.75\textwidth]{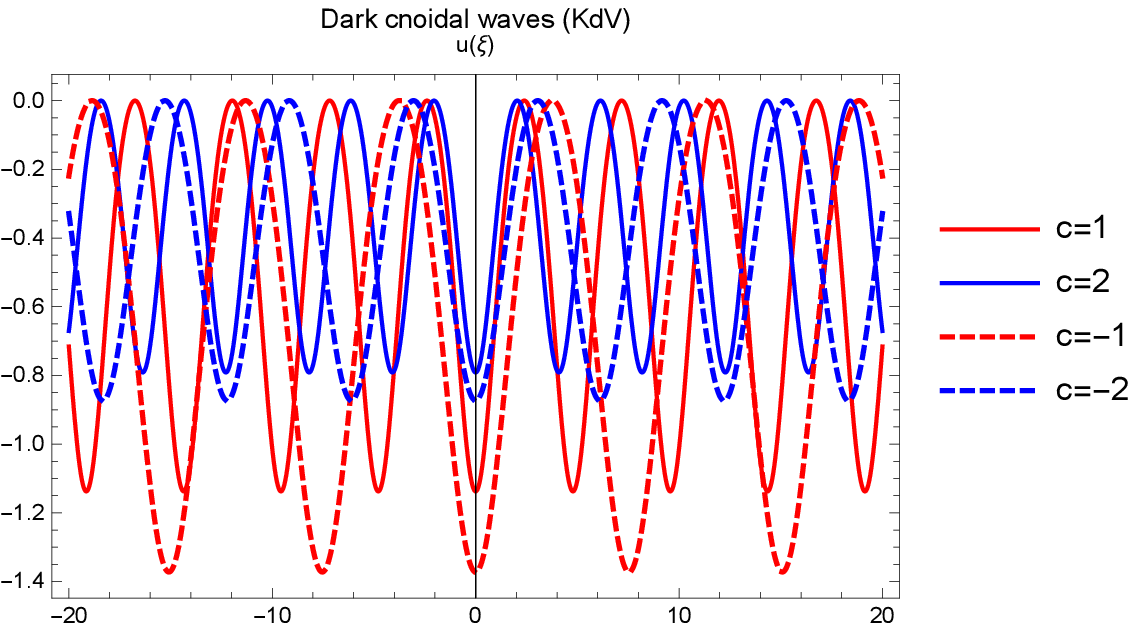}
\caption{Bright cnoidal waves for the KdV  Eq. (\ref{eq:KdV}) with  $\alpha=1,~\kappa=2, ~\mathcal A=1$ (continuous curves) when the waves propagate to the right, and $\alpha=-2,~\kappa=-2 ~\mathcal A=-1$ (dashed curves) when the waves propagate to the left (top panel). Dark cnoidal waves for the KdV  Eq. (\ref{eq:KdV}) with $\alpha=-1,~\kappa=2,  ~\mathcal A=1$ (continuous curves) when the waves propagate to the right, and $\alpha=2,~\kappa=-1, ~\mathcal A=-1$ (dashed curves) when the waves propagate to the left (bottom panel).  }
\label{fig2a}
\end{center}
\end{figure}
When $\mathcal A \rightarrow 0$, and depending on the modulus $m$ we have two extreme cases. First, by  choosing the negative branch of the sqare root $Y_2 \rightarrow 0 \Rightarrow m \rightarrow 0$, and $\mathrm{cn ~\theta} \rightarrow  \cos \theta$, thus cnoidal waves resemble more sinusoidal waves of unchanging shape discovered by Stokes, which in the theory of long waves constitutes a particular case of cnoidal form \cite{Sto,KdV}. Second, by choosing the positive branch of the square root $Y_3 \rightarrow 0 \Rightarrow m \rightarrow 1$, and $\mathrm{cn ~\theta} \rightarrow  \mathrm{sech}~\theta$, thus cnoidal waves loose their periodicity, and  reduce to the solitary waves given by Eq. (\ref{kdv3}).
\section{Kawahara equation ($\beta \ne 0$)}
By using the  second  equation of system (\ref{tan2}) with $B_0=0$, $B_1=0$, $B_2=1$ and $u=Y^2$, the  derivatives of $u$ become
 \begin{equation}\label{sys1}
\begin{array}{l}
u_{\xi\xi}=2({Y_\xi}^2+Y Y_{\xi\xi})\\
u_{\xi\xi\xi}=2(3Y_{\xi}Y_{\xi\xi}+Y Y_{\xi\xi\xi})\\
u_{\xi\xi\xi\xi}=2[3(Y_{\xi\xi})^2+4 Y_{\xi}Y_{\xi\xi\xi}+YY_{\xi\xi\xi\xi}].
\end{array}
\end{equation}
Using the  derivatives of $Y$  from system (\ref{eqa}) in system  (\ref{sys1}), the second and fourth order derivatives of $u$ as function of $Y$  become
 \begin{equation}\label{sys10}
\begin{aligned}
u_{\xi\xi}&=2a_0+3 a_1 Y+4 a_2 Y^2+5 a_3 Y^3\\
u_{\xi\xi\xi\xi}&=\frac 3 2 {a_1}^2+ 8 a_0 a_2+15(a_1a_2+2 a_0a_3)Y+2(8{a_2}^2+21 a_1a_3)Y^2\\
&+65 a_2a_3 Y^3+\frac{105}{2}{a_3}^2 Y^4.
\end{aligned}
\end{equation}
By substituting these derivatives  in Eq. (\ref{eq1aa}) we obtain 
the quartic polynomial in $Y$
\begin{equation}\label{pol}
\sum_{i=0}^4 r_i Y^i \equiv 0,
\end{equation}
with coefficients given by 
 \begin{equation}\label{sys5}
\begin{array}{l}
r_4=\frac 12 (\kappa-105{a_3}^2 \beta)\\
r_3=5a_3(\alpha -13 a_2 \beta)\\
r_2=-c+4a_2 \alpha -16{a_2}^2\beta-42 a_1a_3 \beta\\
r_1=-30 a_0a_3 \beta+3a_1(\alpha-5 a_2 \beta)\\
r_0=-\mathcal A-\frac 3 2{a_1}^2 \beta+2a_0(\alpha-4 a_2\beta).
\end{array}
\end{equation}
Since all coefficients $r_i=0$, by simultaneously  solving  the first four equations of system (\ref{sys5}) we find
 \begin{equation}\label{sys6}
\begin{array}{l}
a_3=\mp\sqrt{\frac{\kappa}{105 \beta}}\\
a_2=\frac{\alpha}{13 \beta}\\
a_1=\sqrt{\frac{5}{21\kappa \beta}}\frac{\mp 36\alpha^2\pm 169c\beta}{2\cdot169 \beta}\\
a_0=\frac{2\alpha(36 \alpha^2-169 c\beta)}{13^3\kappa \beta^2}.
\end{array}
\end{equation}
For real coefficients, and since $\beta>0$,  we require that the  wave steepening  coefficient  $\kappa >0$. By using  these coefficients together  with the last equation of system ($\ref{sys5}$), we find  the integration constant to be
\begin{equation}\label{aaa}
\mathcal A=\frac{(36 \alpha^2-169 \beta c)(2^2\cdot 3^3\cdot 17~ \alpha^2+5\cdot 169 \beta c)}{2^3\cdot7\cdot13^4 \kappa \beta^2}.
\end{equation}
This shows that while for KdV equation the boundary  $\mathcal A$ is {\it arbitrary}, for KE the boundary $\mathcal A$ is {\it fixed} by the system\rq{}s parameters  $\alpha, \beta, \kappa$ and the speed $c$. This means that traveling waves for KE can change shape depending if the speed $c$ is chosen  such a way that $\mathcal A$ is zero or not.
 \subsection{Transmission line equation}
Now we analyze the special case of $\alpha=0$ which  from system (\ref{sys6}) leads to $a_2=a_0=0$ so that  Eq. (\ref{eq:rosenau})  includes a fifth order dispersion term only, and takes the form
\begin{equation}\label{eq:naga}
u_{t}+\kappa uu_{x}-\beta u_{xxxxx}=0.
\end{equation}
 This equation  describes pulses  over a transmission line containing  a large number of LC circuits, and it  occurs by making use of mutual inductance between neighboring  inductors.  It  was first studied by Hasimoto \cite{Has} for shallow water waves near some critical value of surface tension, while  Nagashima \cite{Naga1} performed experiments, and observed the solitary waves using an oscilloscope. Later, in a more general setting,  it was also derived  by Rosenau  using a quasi-continuous formalism that included higher order discrete effects \cite{Ros1,Ros2}.

For $\alpha=0$ Eq. (\ref{eq1aa}) becomes
\begin{equation}\label{eq20aa}
-cu+\frac \kappa 2 u^2 -\beta u_{\xi \xi \xi \xi}  =\mathcal{A},
\end{equation}
so  Eq. (\ref{eq3}) corresponds to 
\begin{equation}\label{eq9}
{Y_{\xi}}^2=Y(a_1+a_3Y^2),
\end{equation}
 which is in fact  a special case of Eq. (\ref{kdv04}) for a real root that has multiplicity two, $Y_2=Y_3$. The reduced  coefficients obtained from system (\ref{sys6})  are given by the expressions
\begin{equation}\label{z0}
\begin{array}{l}
a_3=\mp\sqrt{\frac{\kappa}{105 \beta}} \\
a_1=\pm \frac c2 \sqrt{\frac{5}{21\kappa \beta}},
\end{array}
\end{equation}
and are used to  factor Eq. (\ref{eq9}) to obtain
\begin{equation}\label{fe}
{Y_\xi}^2=-a_3 Y({Y_2}^2-Y^2),
\end{equation}
with ${Y_2}^2=-\frac{a_1}{a_3}=\frac{5c}{2\kappa}$. Since $a_1,a_3$ have opposite signs, then ${Y_2}^2>0$, and because $a_1,a_3 \in \R \Rightarrow \kappa >0$, thus as before, the waves propagate only to the right ($c>0$).
By integrating Eq. (\ref{fe}), and using formula (\ref{can}) we obtain the special cnoidal wave  of constant modulus 
 \begin{equation}\label{can2}
Y(\xi)=Y_2~ \mathrm{cn}^2\left[\frac {\sqrt 2}{2} \sqrt{-a_3Y_2}(\xi-\xi_0);\frac{\sqrt 2}{2}\right].
\end{equation} 
 The wavelength of the waves is  $\lambda=2 \sqrt 2 K\left(\frac {\sqrt 2}{ 2}\right) \sqrt[4]{\frac{42 \beta}{c}}$, where $ K\left(\frac {\sqrt 2}{ 2}\right)=1.85407$, so the periodicity is only a function of the speed $c$ and dispersion $\beta$.  Using the values of $Y_2$, $a_3$, $u=Y^2$ and $\xi_0=0$ the  solution to the transmission line  equation (\ref{eq:naga})  with nonzero boundary condition is
\begin{equation}\label{eq005}
u(x,t)= \frac{5c}{2\kappa} ~\mathrm{cn}^4\left[\frac{\sqrt 2}{2}\sqrt[4]{\frac{c}{42 \beta}}~(x-ct);\frac{\sqrt 2}{2}\right].
\end{equation}
This solution was also obtained  by Yamamoto \cite{Yama7}, Kano \cite{Kano4} and Kudryashov \cite{Kud2},  and represents a train of periodic waves with fixed modulus which tells that the shape is preserved as the pulse travels over the transmission line,  see Fig. \ref{fig4} for the values of $\beta=1,\kappa=2$. 
 \begin{figure}[ht!]
   \begin{center}
\includegraphics[width=0.75\textwidth]{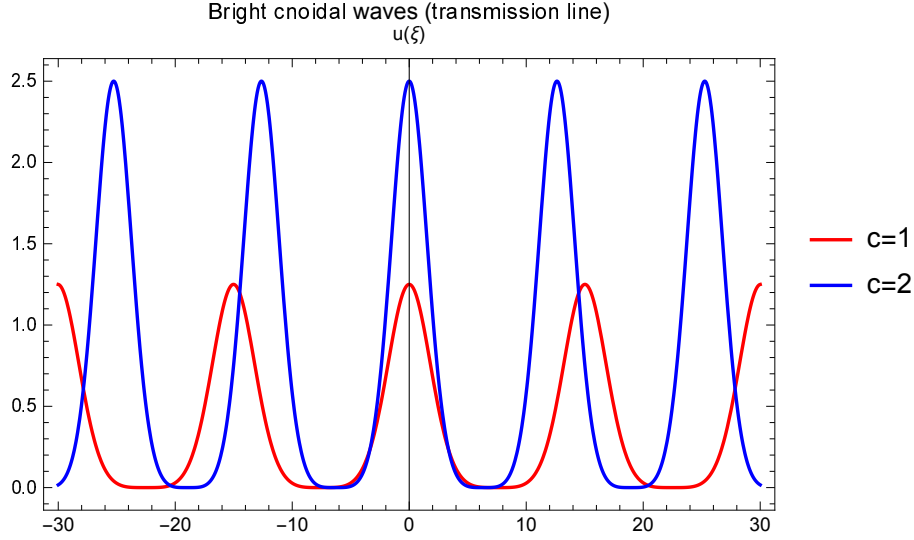}
\caption{Bright cnoidal waves for the transmission line equation  (\ref{eq:naga}) with nonzero boundary conditions  for $\alpha=0, \beta=1, \kappa=2$.}
\label{fig4}
\end{center}
\end{figure} 
\subsection{KE with zero boundary conditions}
Traveling waves with  zero boundary condition $\mathcal A=0$ yield waves which  propagate  with velocities $c=-\frac{2^2\cdot 3^3\cdot 17\alpha^2}{5 \cdot 169 \beta}<0$ or   $c=\frac{36 \alpha^2}{169\beta}>0$. 

{\bf i)} For left  traveling waves, the modular discriminant obtained from the last equation of system (\ref{eq65}) becomes $\Delta=-\frac{6\cdot 1523 \alpha^6}{5^4\cdot 13^6 \beta^6}<0$,  so the polynomial $s_3(t)$ given by Eq. (\ref{eq80}) has non real roots, and because is an  unphysical case it will be omitted. 

{\bf ii)} For right traveling waves,  the system (\ref{sys6}) reduces to 
 \begin{equation}\label{sys6a}
\begin{array}{l}
a_3=\mp\sqrt{\frac{\kappa}{105 \beta}}\\
a_2=\frac{\alpha}{13 \beta}\\
a_1=a_0=0\\
\end{array}
\end{equation}
with corresponding  elliptic equation
\begin{equation}\label{spec}
{Y_\xi}^2=a_3 Y^2\left(Y+\frac{a_2}{a_3}\right),
\end{equation}
Letting $Y\rightarrow -\tilde Y$, and using  solution (\ref{eq00}) we obtain
\begin{equation}\label{eq01}
Y(\xi)= -\tilde Y(\xi)=-\frac{a_2}{a_3}\mathrm{sech}^2\left[\frac 1 2 \sqrt{a_2}(\xi-\xi_0)\right].
\end{equation}
Therefore,  right  traveling wave solution of the  KE  with zero boundary condition and $\xi_0=0$ reduces to
\begin{equation}\label{lab9}
u(x,t)= \frac{105 \alpha^2}{169 \kappa \beta}\mathrm{sech}^4\left[\frac 1 2 \sqrt{\frac{\alpha}{13 \beta}}\left(x-\frac{36 \alpha^2}{169\beta} t\right)\right].
\end{equation}
Since $\beta>0 \Rightarrow \alpha>0$, and since $a_3 \in \R \Rightarrow  \kappa>0$, and  we  obtain bright solitons which only  propagate to the right.  Let us rescale the solution according to the new variables
\begin{equation}\label{sol1}
 \gamma =4 \sqrt{\frac{\kappa}{105 \beta}}, \delta =\frac{4\alpha}{13 \beta} \Rightarrow c=\frac{9 \beta}{4}\delta ^2,
 \end{equation}
 so the solitons take the simpler form
 \begin{equation}\label{x1}
 u(x,t)=\left(\frac{\delta}{\gamma}\right)^2 \mathrm{sech}^4 \left[\frac 1 4 \sqrt{\delta}\left(x-\frac9 4 \delta^2t\right)\right].
 \end{equation}
This  solution  written in the same manner was  also derived  by Yamamoto \cite{Yama7}, Yuan \cite{Yuan0} and Rosenau \cite{Ros2} using  different approaches, and shows once again that the speed of the  wave is proportional to the height and inverse proportional to the width,  as it is the case of the KdV equation, with the difference  that while  for the KdV  the solitons propagate with arbitrary velocity, for KE  they translate with velocity that is fixed by both dispersion coefficients.  In Fig. \ref{fig3} we plot the solitary waves that propagate to the right for the  values of $\alpha =2,\kappa=5$ for which $\gamma=\frac{4}{\sqrt{21 \beta}}$ and $\delta=\frac{8}{13 \beta}$.  Fixing the velocities $c=1,2$ we obtain the values for the dispersion coefficient to  be $\beta=\frac{144}{169}, \frac{72}{169}$.
   \begin{figure}[ht!]
     \begin{center}
\includegraphics[width=0.75\textwidth]{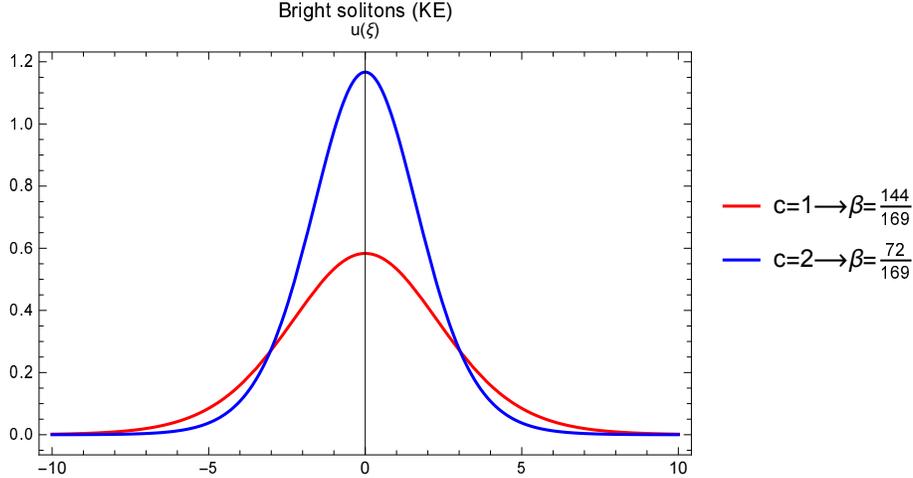}
\caption{Bright solitons for the KE  (\ref{eq:rosenau}) with zero boundary conditions $\mathcal A=0$ and $\alpha=2, \kappa=5$. When $c=1, 2$ then $\beta=\frac{144}{169}, \frac{72}{169}.$ }
\label{fig3}
\end{center}
\end{figure}
\subsection{KE with nonzero boundary conditions}
When all the coefficients from system (\ref{sys6}) are nonzero, we solve Eq. (\ref{eq3})  by reduction to the Weierstrass elliptic equation \cite{Manc}
\begin{equation}\label{eq40}
\wp_{\xi}^2=4 \wp^3-g_2 \wp -g _3,
\end{equation} by the scale-shift linear  transformation 
\begin{equation}\label{eq50}
Y(\xi)=\frac{4}{a_3}\wp(\xi-\xi_0;g_2,g_3)-\frac {a_2}{3a_3}.
\end{equation}
The invariants of the Weierstrass function are given by
\begin{equation}\label{eq60}
\begin{array}{l}
g_2=\frac{{a_2}^2-3 a_1a_3}{12}=2({e_1}^2+{e_2}^2+{e_3}^2)\\
g_3=\frac{3^2 a_1a_2a_3-3^3a_0 {a_3}^2-2{a_2}^3}{16\cdot 3^3}=4e_1e_2e_3,
\end{array}
\end{equation}
and together with the modular discriminant 
\begin{equation}\label{eq70}
\Delta={g_2}^3-3^3 {g_3}^2=16(e_1-e_2)^2(e_1-e_3)^2(e_2-e_3)^2
\end{equation}
are used to classify the solutions of Eq. (\ref{eq3}). The constants $e_i$ are the roots of the cubic polynomial 
\begin{equation}\label{eq80}
s_3(t)=4t^3-g_2t-g_3=4(t-e_1)(t-e_2)(t-e_3)=0,
\end{equation}
and are related to the two periods $\omega_{1,2}$ of the $\wp$ function via the relations $e_{i}=\wp{\left(\frac{\omega_{i}}{2}\right)}$, and $\omega_3=-(\omega_1+\omega_2)$.
Using the constants from system (\ref{sys6})  the invariants and discriminant are
\begin{equation}\label{eq65}
\begin{array}{l}
g_2=\frac{169 \beta c-22 \alpha^2}{2^3\cdot 3^3\cdot 13^3 \beta^2}\\
g_3=\frac{\alpha(3\cdot 169 \beta c-2^7 \alpha^2)}{2^5\cdot 3^3\cdot 5 \cdot 13^3\beta^3}\\
\Delta=\frac{(169 \beta c-36 \alpha^2)(2^2\cdot 3^2\cdot 47\cdot 101\alpha^4+2 \cdot 5^2\cdot 13^4\beta^2 c^2-3\cdot 11 \cdot 13 \cdot 139 \alpha^2 \beta c)}{2^{10}\cdot 3^3\cdot 5^2\cdot 7^3\cdot 13^6\beta^6}.\\
\end{array}
\end{equation}
We now proceed to classify the solutions of Eq. (\ref{eq40}) case-by-case \cite{Steg,Nick}.
\begin{description}
\item [Case (1).] 
We first consider  the degenerate case of $ \Delta = 0\Rightarrow c=\frac{36 \alpha^2}{169 \beta} $  for which  $s_3(t)$  has repeated root $e_i$ of multiplicity two. In this case the  reduced invariants  are
\begin{equation}\label{eq11}
\begin{array}{l}
	g_2=\frac{\alpha^2}{2^2\cdot 3\cdot 169 \beta^2}>0\\
	g_3= -\frac{\alpha^3}{2^3\cdot 3^3\cdot 13^3\beta^3}.\\
\end{array}
\end{equation}
Depending on the sign of $g_3$ we have the sub-cases:

{\bf i)}   $\alpha>0$. By  letting $e_1=e_2=\hat e>0$ with $\hat e=\frac{\alpha}{2^2 \cdot 3\cdot 13 \beta}>0$ then $e_3=-2\hat e=-\frac{\alpha}{78 \beta}<0$, hence
\begin{equation}\label{eq120}
\begin{array}{l}
	g_2=12{\hat e}^2>0\\
	g_3=-8{\hat e}^3<0\\
\end{array}
\end{equation}
 the Weierstrass $\wp$ function is  simplified to
\begin{equation} \label{eq130}
\wp(\xi;12{\hat e}^2,-8{\hat e}^3)=\hat e+3\hat e~ \mathrm{csch}^2(\sqrt{3\hat e}\xi)
\end{equation}
which becomes
\begin{equation} \label{eq140}
\wp(\xi;g_2,g_3)=\frac{\alpha}{2^2 \cdot 3\cdot 13 \beta}\left[1+3 ~\mathrm{csch}^2\left(\frac 1 2 \sqrt{\frac{\alpha}{13 \beta}}\xi\right)\right].
\end{equation}
{\bf ii)}   $\alpha<0$. By  letting $e_2=e_3=-\tilde {e}<0$ with $\tilde {e}=-\frac{\alpha}{2^2 \cdot 3\cdot 13\beta}>0$, then $e_1=2\tilde {e}=-\frac{\alpha}{78 \beta}>0$, hence
\begin{equation}\label{eq1201}
\begin{array}{l}
	g_2=12{\tilde {e}}^2>0\\
	g_3 =8{\tilde {e}}^3>0\\
\end{array}
\end{equation}
 the Weierstrass $\wp$ function is  simplified to
\begin{equation} \label{eq1301}
\wp(\xi;12{\tilde {e}}^2,8{\tilde {e}}^3)=-{\tilde {e}}+3{\tilde {e}}~ \mathrm{csc}^2(\sqrt{3{\tilde {e}}}\xi)
\end{equation} which is
\begin{equation} \label{eq1401}
\wp(\xi;g_2,g_3)=\frac{\alpha}{156 \beta}\left[1-3 ~\mathrm{csc}^2\left(\frac 1 2 \sqrt{\frac{-\alpha}{13 \beta}}\xi\right)\right].
\end{equation}
Using the transformation  (\ref{eq50}) together with $u=Y^2$, the solutions are
\begin{equation}\label{Tab2}
u(x,t) =\left\{
     \begin{array}{lr}
       \frac{105 \alpha^2}{169 \kappa \beta}\mathrm{csc}^4\left[\frac 1 2 \sqrt{\frac{-\alpha}{13 \beta}}\left(x-\frac{36 \alpha^2}{169\beta} t\right)\right]&\quad\quad;\alpha<0\\
       \frac{105 \alpha^2}{169 \kappa \beta}\mathrm{csch}^4\left[\frac 1 2 \sqrt{\frac{\alpha}{13 \beta}}\left(x-\frac{36 \alpha^2}{169\beta} t\right)\right]&\quad\quad ;\alpha>0
     \end{array}
   \right.
\end{equation}
 Notice that these solutions are  the second linearly independent set  obtained from the reduction of the $\wp$ function corresponding exactly to the solutions of KE with zero boundary conditions of  (\ref{lab9}), since for $c=\frac{36 \alpha^2}{169 \beta}\Rightarrow \mathcal A=0$. Because  these functions are unbounded, the traveling waves  are unphysical  so they will also be omitted.

\item [Case (2).] If $\Delta\ne 0$ we find traveling waves with arbitrary velocity and we include   two  particular solutions which will fix the velocities  of the traveling waves as functions of  dispersion coefficients as follows:  the equianharmonic ($g_2=0$),  and lemniscatic case ($g_3=0$) respectively.

{\bf i)} For general  solution $g_2 \ne0 $, $g_3\ne0$  the waves travel with arbitrary velocity $c$  and the solutions may be reduced  in a manner similar to the simplification for the lemniscatic  case below.

{\bf ii)} For the equianharmonic case $g_2=0\Rightarrow c=\frac{22 \alpha^2}{169 \beta}$, thus $g_3=-\frac{31 \alpha^3}{2^4\cdot 3^3\cdot 5\cdot 13^3 \beta^3}$, and discriminant is $\Delta=-27 {g_3}^2<0$ with solution to Eq. (\ref{eq40}) given by 
\begin{equation}\label{sol3}
\wp(\xi-\xi_0;0,g_3)=\wp\left(\xi-\xi_0;0,-\frac{31 \alpha^3}{2^4\cdot 3^3\cdot 5\cdot 13^3  \beta^3}\right).
\end{equation} 
Since $\Delta<0$ then the polynomial $s_3(t)$ given by Eq. (\ref{eq80}) has non real roots, and this case will be omitted being unphysical as well. Using the transformation (\ref{eq50}) together with $u=Y^2$, the general and equianharmonic  solutions are
\begin{equation}\label{Tab5}
u(x,t) =\left\{
     \begin{array}{lr}
          \frac{35}{3\cdot 13^2\kappa \beta}\left[\alpha-2^2 \cdot 3\cdot 13 \beta\wp\left(x-c t;\frac{13^2 c\beta-22 \alpha^2}{2^3\cdot 3^3\cdot 13^3  \beta^2},\frac{\alpha(3\cdot 13^2\beta c-2^57\alpha^2)}{2^5\cdot 3^3\cdot 5 \cdot 13^3\beta^3}\right)\right]^2 \\

      \frac{35}{3\cdot 13^2\kappa \beta}\left[\alpha-2^2 \cdot 3\cdot 13\beta\wp\left(x-\frac{22 \alpha^2}{13^2 \beta} t;0,-\frac{31 \alpha^3}{2^4\cdot 3^3\cdot 5\cdot 13\beta^3}\right)\right]^2.
     \end{array}
     \right.
\end{equation}

{\bf iii)} For the lemniscatic case $g_3=0\Rightarrow c=\frac{2^7 \alpha^2}{3\cdot 169 \beta}$, thus $g_2=\frac{31 \alpha^2}{2^2\cdot 3^2\cdot 7\cdot 169 \beta^2}>0$, and discriminant is $\Delta={g_2}^3>0$, with solution to Eq. (\ref{eq40}) given by
\begin{equation}\label{sol4}
\wp(\xi-\xi_0;g_2,0)=\wp\left(\xi-\xi_0;\frac{31 \alpha^2}{2^2\cdot 3^2\cdot 7\cdot 13^2  \beta^2},0\right).
\end{equation}
In this case the roots of $s_3(t)$ are real and are given by $e_3=-\frac{\sqrt {g_2}}{2}$, $e_2=0$, $e_1=\frac{\sqrt {g_2}}{2}$. Although the Weierstrass unbounded function given by (\ref{sol4}) has poles aligned on the real axis of the $\xi-\xi_0$ complex plane, we can choose $\xi_0$ in such a way to shift these poles a half of period above the real axis, so that the elliptic function simplifies using the formula \cite{Kano4,Whitt}
\begin{equation}\label{sol5}
\wp(\xi;g_2,0)=e_3+(e_2-e_3)\mathrm{sn}^2[\sqrt{e_1-e_3}(\xi-\xi_0\rq{});m]
\end{equation}  with elliptic modulus $m=\sqrt{\frac{e_2-e_3}{e_1-e_3}}.$
Using the values of the roots together with $\xi_0\rq{}=0$ we obtain
\begin{equation}\label{sol6}
\wp(\xi;g_2,0)=-\frac{\sqrt{g_2}}{2}\mathrm{cn}^2\left[\sqrt[4]{g_2}\xi;\frac{\sqrt{2}}{2}\right],
\end{equation}
thus, the solutions for the lemniscatic case are reduced using the transformation (\ref{eq50}) to cnoidal waves, and they become
\begin{equation}\label{Tab4}
u(x,t) =
     \begin{array}{lr}
        \frac{5\alpha^2}{3\cdot 7 \cdot 13^2 \kappa \beta}\left\{7+\sqrt{7 \cdot 31}\mathrm{cn}^2\left[\sqrt[4]{\frac{31}{7}}\sqrt{\frac{\alpha}{78 \beta}}\left(x-\frac{2^7 \alpha^2}{3\cdot 13^2\beta}t\right);\frac{\sqrt 2}{2}\right]\right\}^2
     \end{array}
\end{equation}
In Fig. \ref{fig5} we present cnoidal wave-trains for the lemniscatic case with nonzero boundary conditions, and coefficients $\alpha=1, \kappa=1$. Fixing the velocities to $c=1,2$  the fifth order dispersion coefficient  is $\beta=\frac{128}{507},\frac{64}{507}$.
  \begin{figure}[ht!]
  \begin{center}
\includegraphics[width=0.75\textwidth]{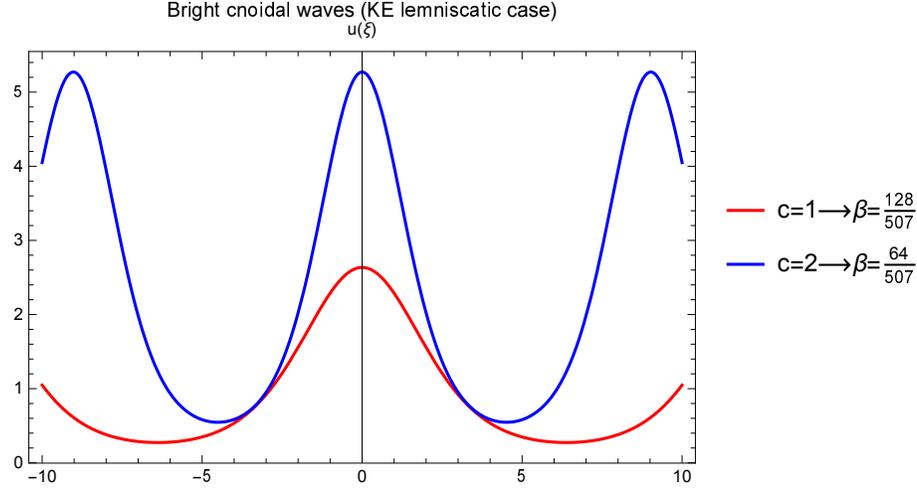}
\caption{Bright cnoidal waves for the KE  (\ref{eq:rosenau}) with nonzero boundary conditions for  $\alpha=1, \kappa=1$. When $c=1, 2$ then $\beta=\frac{128}{507}, \frac{64}{507}.$ }
\label{fig5}
\end{center}
\end{figure}
\end{description}

\section{Conclusion}
In this paper we applied the  generalized elliptic function method to find traveling wave solutions to  Kawahara, transmission line, and Korteweg-de Vries equations, which has the advantage of finding  solutions to nonlinear evolution equations as polynomial combinations of elliptic functions. Depending on the boundary conditions, these solutions can be reduced, if necessary, to the hyperbolic, periodic, or Jacobian elliptic functions. 

By assuming a polynomial ansatz that  satisfies an elliptic equation,  in the case of  KdV equation, we find the well-known solitary waves,  as well as  wave-trains of  cnoidal waves which are either compressive  or rarefactive that  propagate  in both directions with arbitrary velocity.

In the case of KE  the traveling wave  solutions  are written in terms of Weierstrass elliptic functions  which can  be reduced to the hyperbolic (for zero boundary conditions) or Jacobi elliptic functions (for nonzero boundary conditions). For the general case the Weierstrass elliptic functions are unbounded, while for the lemniscatic case, they reduce to periodic cnoidal waves.

 While  for the KdV  equation the solitary waves that are both compressive and rarefactive   propagate with arbitrary velocity, for KE  only compressive  waves are found that propagate to one direction with a velocity that depends on both dispersion coefficients. 

\begin{acknowledgements}
The author  would like to acknowledge Professor  Pisin Chen from  Leung Center for Cosmology and Particle Astrophysics (LeCosPA) for support during his stay in Taipei, and  Professors Juan-Ming Yuan and Haret C. Rosu for their helpful comments and discussions on Kawahara equation. 
\end{acknowledgements}

\section{Appendix}

\begin{lemma}
Traveling wave solutions to KE (\ref{eq:rosenau}), satisfy  only the Weierstrass elliptic equation with cubic nonlinearity.
\end{lemma}
\begin{proof}
 Letting 
\begin{equation}\label{mew}
{Y_\xi}^2=\sum_{i=0}^m a_i Y^i, \quad \quad a_M \ne 0.
\end{equation}
For KdV Eq. (\ref{eq:KdV}) $u=Y$  and the leading term for the second order  derivative term is $u_{\xi\xi} =\frac 1 2 m a_m Y^{m-1}$. By matching the terms $u^2$ with  $u_{\xi\xi}$ in Eq. (\ref{eq1aaa}) we obtain $m=3$. For KE (\ref{eq:rosenau}) $u=Y^2$ and  the  leading term for the fourth order  derivative term is $u_{\xi\xi\xi\xi}=\frac 1 2 m(m+2)(3m-2) {a_m}^2 Y^{2(m-1)}$. By matching the terms $u^2$ and $u_{\xi\xi\xi\xi}$ in Eq. (\ref{eq1aa}) we also obtain $m=3$. Notice that this method will not work if an ODE contains both even and odd derivative terms. 
\end{proof}\eproof



\end{document}